\documentclass[letterpaper, 10 pt, conference]{ieeeconf}  

\IEEEoverridecommandlockouts                              
\usepackage{graphicx} 
\usepackage{epsfig} 
\usepackage{times} 
\usepackage{amsmath} 
\usepackage{amssymb}  
\usepackage{blindtext}
\usepackage{algpseudocode}
\usepackage{hyperref}
\usepackage{xcolor}
\usepackage{cleveref}
\usepackage{booktabs}
\usepackage{multirow}

\usepackage[fancythm,fancybb,morse,draft]{jphmacros2e}
\renewcommand{\P}[1]{\mathsf{P}_{#1}}

\title{\LARGE \bf
Zero-sum turn games using $Q$-learning:\\ finite computation with security guarantees}

 \author{Sean Anderson, Chris Darken, and Jo\~ao P. Hespanha
 	\thanks{© 2025 IEEE. Personal use of this material is permitted. 
 		Permission from IEEE must be obtained for all other uses.}
            \thanks{This material is based upon work supported by the National Science Foundation Graduate Research Fellowship under Grant No. 2139319 and by the U.S. Office of Naval Research MURI grant No. N00014-23-1-2708 and No. N00014-24-1-2492. Any opinion, findings, and conclusions or recommendations expressed in this material are those of the authors(s) and do not necessarily reflect the views of the National Science Foundation or U.S. Office of Naval Research.}
                       \thanks{S. Anderson and J. P. Hespanha are with the Electrical and Computer Engineering Department, University of California, Santa Barbara, CA 93106 USA (e-mail: seananderson@ucsb.edu, hespanha@ece.ucsb.edu). C. Darken is with the Computer Science Department, Naval Postgraduate School, Monterey, CA 93943 USA (email: cjdarken@nps.edu).}}

\begin{document}

\maketitle
\thispagestyle{empty}

\begin{abstract}
  This paper addresses zero-sum ``turn'' games, in which only one
  player can make decisions at each state.  We show that pure
  saddle-point state-feedback policies for turn games can be
  constructed from dynamic programming fixed-point equations for a
  single value function or $Q$-function. These fixed-points can be
  constructed using a suitable form of $Q$-learning. For discounted
  costs, convergence of this form of $Q$-learning can be established
  using classical techniques. For undiscounted costs, we provide a
  convergence result that applies to finite-time deterministic games,
  which we use to illustrate our results. For complex games, the
  $Q$-learning iteration must be terminated before exploring the
  full-state, which can lead to policies that cannot guarantee the
  security levels implied by the final $Q$-function. To mitigate this,
  we propose an ``opponent-informed'' exploration policy for selecting
  the $Q$-learning samples. This form of exploration can guarantee
  that the final $Q$-function provides security levels that hold, at
  least, against a given set of policies. A numerical demonstration for a multi-agent
  game, Atlatl, indicates the effectiveness of these methods.
\end{abstract}

\section{Introduction} \label{sec:introduction}

Learning in game settings has been successfully applied in highly
complex and large-scale games such as Atari in 2013 via deep
$Q$-learning networks~\cite{mnih_playing_2013} or Go with AlphaGo
exceeding human capabilities in 2016~\cite{silver_mastering_2016}. We
focus this paper on $Q$-learning, a method originally developed in
\cite{watkins_learning_1989} that extends the dynamic programming
principle to model-free reinforcement learning by eliminating the need
for a known transition model.

This paper is focused on zero-sum games where players play in
``turns,'' meaning that, at each state of the game, only one player
can make decisions, while the other player waits for a turn. Many
board-games (like chess, Connect Four, or Tic-Tac-Toe) are ``turn''
games with alternating moves, but we consider here a general form of
turn games where players do not necessarily alternate, and the same
player can make several ``moves'' in a row, before the other player
gets a turn. In fact, it may happen that one player can only make
decisions at the start of the game, and rapidly lose the ability to do
so. The motivation for this setup are games of strategy, like Atlatl
\cite{cjdarken_public_2025,rood2022scaling}, Mancala, and Dots and Boxes. Further
motivations include pursuit-evasion problems where the evader loses
the ability to act while cornered or in network security where an
attacker may control certain states (e.g., when host is vulnerable)
and a defender can respond at select times (e.g., publish
patches). Although such games can be modeled, in principle, as
alternating play by chaining multiple actions into an ``extended''
single action, this approach leads to an exponential growth of the
action space. By treating them directly as turn-based, we avoid this
combinatorial blow-up, which reduces the computational complexity and
memory requirements for $Q$-learning. As discussed in
Section~\ref{sec:prob_form}, we model zero-sum turn games as controlled Markov chains.

The solution concept we pursue are saddle-point state-feedback
policies, for which none of the players regret the choice of a
policy in view of the choice of the other player. Such policies come
with ``security values,'' in the sense of guaranteed levels of reward
for each player~\cite{Hespanha17}. As opposed to general zero-sum
Markov games (see, e.g., \cite{FilarVrieze97}), turn games with finite
state and action spaces have pure saddle-point policies under very
mild assumptions, avoiding the need to consider mixed or behavioral
policies. In fact, we show in Section~\ref{sec:DP} that pure
saddle-point state-feedback policies can be constructed from a dynamic
programming fixed-point equation that resembles that associated with
single-agent optimizations, but with an additional term that captures the
zero-sum turn nature of the game.

The computational method proposed to find saddle-point policies is
directly inspired by single-agent $Q$-learning, adapted to the
turn-games setting (see Section~\ref{sec:q-learning}). This enables
the use of classical techniques to prove convergence of the
$Q$-learning iteration. We show this to be the case in deterministic
games with a finite time-horizon.

We are especially interested in games with very large state spaces,
for which the $Q$-learning iteration must be terminated at some finite
iteration $K$, for which only a (potentially very small) fraction of
the state-space has been explored. We show through formal arguments
(and later through simulations) how this can lead to policies that
are not security policies and, in fact, can lead to outcomes much
worse that the value of the $Q$ function at iteration $K$ would
indicate. Towards mitigating this problem, we show in
Section~\ref{sec:fragility} that it is possible to ``protect'' the
$Q$-learning policies against a known set of policies for the opponent
by appropriate exploration of the state-space during the $Q$-learning
iterations. This result is used in Section~\ref{sec:exploration} to
construct an ``opponent-informed'' exploration policy for selecting
the $Q$-learning samples. This form of exploration can guarantee that
the final $Q$-function provides security levels that hold, at least,
against a given set of policies.

The results and algorithms proposed are illustrated in
Section~\ref{sec:numerical-results} in the sophisticated Atlatl
strategy game~\cite{cjdarken_public_2025,rood2022scaling}, which allows us to scale the
results to state-spaces for which a full exploration is
computationally prohibitive with modest computation.

\subsection{Related work}

\textit{Single-player} $Q$-learning was originally studied in the
single player context~\cite{watkins_learning_1989}. Convergence proofs
were formally presented in~\cite{watkins_technical_1992} for the
discounted and undiscounted cases in stochastic settings. The
convergence analysis was extended in
\cite{tsitsiklis1994asynchronous} based on stochastic
approximation arguments and was the first to include contraction
arguments for $Q$-learning.

Early work on \textit{Markov games} showed the existence and
characterization of equilibria in matrix games as well as the
convergence of value and policy
iteration~\cite{shapley_stochastic_1953}. We refer the reader to the
monograph~\cite{FilarVrieze97} for a systematic treatment of zero-sum
Markov games. The authors in~\cite{littman_generalized_1996} showed
convergence in general Markov games under contraction
assumptions. Turn-based games are less common and have appeared in
\cite{sidford_solving_2019, jia_feature-based_2019} with discounted
$Q$-learning and in~\cite{shah_reinforcement_2020} with a discounted
Explore-Improve-Supervise scheme.

The authors in~\cite{brafman2002r} propose R-max, a model-based
reinforcement learning algorithm for stochastic games, including the
zero-sum case, and characterize its \emph{finite sample performance}
to achieve a policy that yields an expected reward within a tolerance
of the associated value function. For discounted alternating play
zero-sum stochastic games, \cite{sidford_solving_2019} provides a
near-optimal time and sample complexity result. Additional work, based
on embedding the state-transition function in a feature space, allows
for near-optimal strategies with finite sample complexity when using
$Q$-learning~\cite{jia_feature-based_2019}.

\emph{Robust and adversarial reinforcement learning} share
connections with zero-sum games, particularly using simultaneous play. We
omit a detailed comparison, but work on robust Markov decision processes
can be found in~\cite{nilim_robust_2005,iyengar_robust_2005} and
appears in earlier work on $Q$-learning
\cite{littman_generalized_1996}.

The idea of using \emph{externally provided (sets of) policies for training purposes} has appeared previously in areas such as Policy Space Response Oracles (PSRO)~\cite{lanctot2017unified} and fictitious self-play~\cite{vinyals2019grandmaster}, differing from our work in that these sets are typically constructed iteratively.

\section{Zero-sum turn games}   \label{sec:prob_form}

In \emph{turn games}, only one player can make a decision at each
state. For 2-player turn games, this means that the state-space
$\scr{S}$ can be partitioned as the union of two disjoint sets
$\scr{S}_1\cap\scr{S}_2=\emptyset$ with the understanding that, when the state
$s_t$ at time $t\ge 1$ belongs to $\scr{S}_1$, the following action
$a_t\in\scr{A}$ is selected by player $\P1$, otherwise
$s_t\in\scr{S}_2$ and the action is selected by player
$\P2$. Single-player optimizations correspond to the special case
$\scr{S}_2=\emptyset$, which means that only $\P1$ makes decisions.

We consider here \emph{zero-sum turn games} with symmetric rewards for
the two players and denote by $r_{t+1}\in\scr{R}\subset\R$, $t\ge 1$ the
immediate reward collected by the player that selected the action
$a_t$ at time $t$. The total reward collected by player $\P{i}$,
$i\in\{1,2\}$ for the initial state $s_1\in\scr{S}$ is then given by
\begin{align}\label{eq:reward-i}
  J_i(s_1) \eqdef \sum_{t=1}^\infty \E [r_{t+1} \sgn_i(s_t)]
\end{align}
where $\sgn_i(s_t)=1$ if $s_t \in \scr{S}_i$ and $\sgn_i(s_t)=-1$
otherwise.
The sets $\scr{S},\scr{A},\scr{R}$ are assumed finite and the state
$s_t$ is a \textit{stationary controlled Markov chain} in the sense
that
\begin{multline}
  P(s_{t+1} = s', r_{t+1} = r \mid s_t = s, a_t =a, \cal{F}_t) \\
  = p(s', r \mid s, a), \quad
  \forall t \geq 1,\; s, s' \in \scr{S},\; a \in \scr{A},\; r \in \scr{R}
\end{multline}
for a transition/reward probability function
$p: \scr{S} \times \scr{R} \times \scr{S} \times \scr{A} \rightarrow [0, 1]$ that satisfies
\begin{align}
  \sum_{s' \in \scr{S}, r \in \scr{R}} p(s', r \mid s, a) = 1, \quad s \in \scr{S},\; a \in \scr{A},
\end{align}
and $\cal{F} \eqdef \{\cal{F}_t : t \geq 1\}$ denotes the filtration
generated by $\{(s_t, a_t, r_t): t \geq 1\}$.
We say that a game is \emph{deterministic} if $p(\cdot,\cdot)$ only takes
values in the set $\{0,1\}$ and
that the game \emph{terminates in finite time} if there exists a
finite time $T\ge 1$ such that $r_t=0$, $\forall t\ge T$ with probability one,
regardless of the actions $a_t\in\scr{A}$ selected.

%
\subsection{Value function for turn games}

A policy for a player $\P{i}$, $i \in \{1, 2\}$ is a map
$\pi_i: \scr{S}_i \rightarrow \scr{A}$ that selects the action
$a_t= \pi_i(s_t)$ when the state $s_t$ is in $\scr{S}_i$. For a pair of
policies $(\pi_1, \pi_2)$ for players $\P{1},\P{2}$, respectively, we
define the \emph{policies' value function} as
\begin{align}\label{eq:reward-sum}
  V_{\pi_1,\pi_2}(s)
  &\eqdef \sgn_1(s) \sum_{t=\tau}^\infty \E_{\pi_1, \pi_2} [r_{t+1}\sgn_1(s_t) \mid s_\tau = s] \notag\\
  &= \sgn_2(s) \sum_{t=\tau}^\infty \E_{\pi_1, \pi_2} [r_{t+1}\sgn_2(s_t) \mid s_\tau = s],
\end{align}
where the subscript in the expected values highlights that the actions
are determined by the given policies. The second equality is a
consequence of the fact that $\sgn_1(s) =- \sgn_2(s)$,
$\forall s\in\scr{S}$ and the time $\tau\ge1$ from which the summation is started
does not affect its value due to the stationarity of the Markov chain.
The reward~\eqref{eq:reward-i} collected by player $\P{i}$ can be
obtained from the value function using
\begin{align}\label{eq:J-V-i}
  J_i(s_1)=\sgn_i(s_1)V_{\pi_1,\pi_2}(s_1), \quad\forall i\in\{1,2\}.
\end{align}


\subsection{Saddle-points and security policies}\label{sec:equilibria}
A pair of policies $(\pi_1^*, \pi_2^*)$, for players $\P{1}, \P{2}$
respectively, is a \emph{(pure) feedback saddle-point} if for every
other policies $(\pi_1,\pi_2)$ for players $\P{1}, \P{2}$, respectively,
we have that
\begin{subequations}\label{eq:saddle-point}
  \begin{align}
    \label{eq:saddle-point-1}
    J_1^*(s)&\eqdef\sgn_1(s)V_{\pi_1^*, \pi_2^*}(s) \geq \sgn_1(s) V_{\pi_1, \pi_2^*}(s), \\
    \label{eq:saddle-point-2}
    J_2^*(s)&\eqdef\sgn_2(s)V_{\pi_1^*, \pi_2^*}(s) \geq \sgn_2(s) V_{\pi_1^*, \pi_2}(s),
  \end{align}
\end{subequations}
$\forall s\in\scr{S}$, and $J_1^*(s_1)=-J_2^*(s_1)$ is called the \emph{value
  of the game}. In view of~\eqref{eq:J-V-i}, \eqref{eq:saddle-point-1}
expresses no regret in the sense that $\P{1}$ does not regret its
choice of $\pi_1^*$ (over any other policy $\pi_1$) against $\pi_2^*$ and,
similarly, \eqref{eq:saddle-point-2} expresses no regret for $\P{2}$.
Saddle-point policies for zero-sum games are known to also be
\emph{security policies with values $J_1^*(s_1)$ and $J_2^*(s_1)$ for
  players $\P{1}$ and $\P{2}$, respectively}; in the sense
that
\begin{subequations}\label{eq:security-policy}
  \begin{align}
    J_1^*(s_1)
    &=\max_{\pi_1} \min_{\pi_2} \sgn_1(s_1)V_{\pi_1,\pi_2}(s_1)\notag\\
    &=\min_{\pi_2} \sgn_1(s_1)V_{\pi_1^*,\pi_2}(s_1)\\
    J_2^*(s_1)
    &=\max_{\pi_2} \min_{\pi_1} \sgn_2(s_1)V_{\pi_1,\pi_2}(s_1)\notag\\
    &=\min_{\pi_1} \sgn_2(s_1)V_{\pi_1,\pi_2^*}(s_1),
  \end{align}
\end{subequations}
which means that, by using the policy $\pi_i^*$, the player $\P{i}$ can
expect a reward at least as large as $J_i^*(s_1)$, \emph{no matter
  what policy the other player uses}~\cite{Hespanha17}.

\section{Dynamic Programming}\label{sec:DP}

The value function~\eqref{eq:reward-sum} satisfies a fixed-point
equation that resembles that of single-player games, but adapted to
the setting of zero-sum turn games. We present below a sufficient
fixed-point condition for a function $V$ to
satisfy~\eqref{eq:reward-sum}. This result requires the following
definition: we say that a function $V: \scr{S} \rightarrow \mathbb{R}$ is
\emph{absolutely summable} if for every pair of policies
$(\pi_1, \pi_2)$ for players $\P{1},\P{2}$, respectively, the series
\begin{align}
  \sum_{t=\tau}^\infty \E_{\pi_1,\pi_2} [V(s_t) \mid s_\tau = s],
\end{align}
is absolutely convergent $\forall s \in \scr{S}, \tau \geq 1$. For games that
terminate in finite time $T$, any function
$V: \scr{S} \rightarrow \mathbb{R}$ for which $V(s_t)=0$,
$\forall t\ge T$ with probability one is absolutely summable: the series
degenerates into a finite summation.

\begin{lemma}[Saddle-point sufficient condition]
  \label{lemma:dp-saddle}
  Suppose there exists an absolutely summable function
  $V: \scr{S} \rightarrow \mathbb{R}$ for which
  \begin{multline}\label{eq:bellman-saddle-value}
    V(s) = \max_{a \in \scr{A}} \E \big[r_{t+1} + \sgn_1(s_t) \sgn_1(s_{t+1})V(s_{t+1}) \\
    \mid s_t = s, a_t = a\big], \quad \forall s \in \scr{S}, t \geq 1,
  \end{multline}
  then for any pair of policies $(\pi_1^*, \pi_2^*)$ for which
  \begin{multline}
    \label{eq:bellman-policy}
    \pi_i^*(s)
    \in \argmax_{a \in \scr{A}} \E \big[r_{t+1}+ \sgn_1(s_t)\sgn_1(s_{t+1}) V(s_{t+1}) \\
    \mid s_t = s, a_t = a\big], \quad \forall s \in \scr{S}_i, i \in \{1, 2\}
  \end{multline}
  form a feedback saddle-point and are also security policies, with
  value $J_1^*(s_1)=\sgn_1(s_1)V(s_1)=-J_2^*(s_1)$. \frqed
\end{lemma}

We can recognize in~\eqref{eq:bellman-saddle-value} a modified version
of the usual dynamic programming equation for single-agent
optimization, with a new term $\sgn_1(s_t) \sgn_1(s_{t+1})$ that
captures the turn nature of the game.

The following proposition is needed to prove
Lemma~\ref{lemma:dp-saddle} and establishes a few relationships
between fixed-point equalities and inequalities and the value function
in~\eqref{eq:reward-sum}.
\begin{proposition}
  \label{prop:dp-ineq}
  Suppose there exists an absolutely summable function
  $V: \scr{S} \rightarrow \mathbb{R}$ for which
  \begin{multline}\label{eq:DP-policy-sufficiency}
      V(s) = \E_{\pi_1, \pi_2}[r_{t+1} + \sgn_1(s_t) \sgn_1(s_{t+1}) V(s_{t+1})\\
      \mid s_t = s], \quad \forall s \in \scr{S},\; t \geq 1,
  \end{multline}
  then $V(s)$ must be equal to the value function
  in~\eqref{eq:reward-sum}.
  Instead, if
  \begin{align}\label{eq:DP-policy-inequality}
    \sgn_1(s)V(s)
    &\le\E_{\pi_1,\pi_2}\big[ \sgn_1(s_t)r_{t+1}
    +\sgn_1(s_{t+1})V(s_{t+1})
    \notag\\&\qquad
    \;\big|\;s_t=s\big],
    \quad \forall s\in\scr{S},\;t\ge 1.
  \end{align}
  then
  \begin{align}
    \label{eq:V-2-inequality}
    \sgn_1(s)V(s)
    &\le \sgn_1(s)V_{\pi_1,\pi_2}(s)
    \quad \forall s\in\scr{S},\; \tau\ge 1
  \end{align}
  and, if instead,
  \begin{align}\label{eq:DP-policy-inequality-2}
    \sgn_1(s)V(s)
    &\ge\E_{\pi_1,\pi_2}\big[ \sgn_1(s_t)r_{t+1}
    +\sgn_1(s_{t+1})V(s_{t+1})
    \notag\\&\qquad
    \;\big|\;s_t=s\big],
    \quad \forall s\in\scr{S},\;t\ge 1.
  \end{align}
  then
  \begin{align}
    \label{eq:V-2-inequality-2}
    \sgn_1(s)V(s)
    &\ge\sgn_1(s)V_{\pi_1,\pi_2}(s)
    \quad \forall s\in\scr{S},\; \tau\ge 1.
  \end{align}
\end{proposition}

\begin{proof-proposition}{\ref{prop:dp-ineq}}
  Multiplying both sides of~\eqref{eq:DP-policy-sufficiency} by
  $\sgn_1(s)$, we conclude that
  \begin{multline}\label{eq:sgn-equal}
    \sgn_1(s) V(s) - \E[\sgn_1(s_{t+1})V(s_{t+1}) \mid s_t=s] \\
    = \E[r_{t+1}\sgn_1(s_t) \mid s_t=s],
  \end{multline}
  $\forall s \in \scr{S}$. Using this equality with $s=s_t$ and taking conditional
  expectations given $s_\tau=\bar{s} \in \scr{S}$, for some
  $\tau \leq t$, we conclude that
  \begin{multline*}
    \E [ \sgn_1(s_t)V(s_t) - \sgn_1(s_{t+1})V(s_{t+1}) \mid s_\tau=\bar{s}] \\
    = \E[r_{t+1}\sgn_1(s_t) \mid s_\tau=\bar{s}].
  \end{multline*}
  Adding both sides of this equality from $t=\tau$ to $t \to \infty$ and using
  absolute convergence of the series on the right-hand side, we
  conclude that
  \begin{align}\label{eq:DP-policy-sufficiency-mult-sum}
    \sgn_1(\bar{s}) V(\bar{s})
    =\sum_{t = \tau}^{\infty} \E [ r_{t+1} \sgn_1(s_t) \mid s_\tau=\bar{s}],
  \end{align}
  from which~\eqref{eq:reward-sum} follows by multiplying both sides
  of the equality above by $\sgn_1(\bar{s})$.
  
  If instead of an equality in~\eqref{eq:DP-policy-sufficiency}, we
  have an inequality, we start the reasoning above with an inequality
  instead of~\eqref{eq:sgn-equal}, which eventually leads to a similar
  inequality, instead of the equality
  in~\eqref{eq:DP-policy-sufficiency-mult-sum}, from which we
  obtain~\eqref{eq:V-2-inequality}
  or~\eqref{eq:V-2-inequality-2}.\frQED
\end{proof-proposition}

\begin{proof-lemma}{\ref{lemma:dp-saddle}}
  When $s \in \scr{S}_1$, the policy $\pi_1^*(s)$ reaches the maximum
  in~\eqref{eq:bellman-saddle-value}, but an arbitrary policy
  $\pi_1(s)$ may not and therefore
  \begin{align}
    \label{eq:subopt-pi}
    V(s)
    &= \E [r_{t+1} + \sgn_1(s_t) \sgn_1(s_{t+1}) V(s_{t+1})
    \notag\\&\qquad
    \mid s_t = s, a_t = \pi_1^*(s)]  \notag\\
    &= \E_{\pi_1^* \pi_2} [ r_{t+1} + \sgn_1(s_t) \sgn_1(s_{t+1})
    V(s_{t+1})
    \notag\\&\qquad
    \mid s_t = s]
    \\
    \label{eq:subopt-pi-ineq}
    &\geq \E[r_{t+1} + \sgn_1(s_t) \sgn_1(s_{t+1}) V(s_{t+1})
    \notag\\&\qquad
    \mid s_t = s, a_t = \pi_1(s)] \notag\\
    &= \E_{\pi_1\pi_2}[r_{t+1} + \sgn_1(s_t) \sgn_1(s_{t+1}) V(s_{t+1})
    \notag\\&\qquad
    \mid s_t = s], \quad \forall s \in \scr{S}_1, \forall \pi_1, \pi_2.
  \end{align}
  In contrast, when $s \in \scr{S}_2$, the policy $\pi_2^*(s)$ reaches the
  maximum in~\eqref{eq:bellman-saddle-value}, but an arbitrary policy
  $\pi_2(s)$ may not and we have
  \begin{align}
    \label{eq:subopt-pi-2}
    V(s)
    &= \E_{\pi_1\pi_2^*}[r_{t+1} + \sgn_1(s_t) \sgn_1(s_{t+1}) V(s_{t+1})
    \notag\\&\qquad
    \mid s_t = s] \\
    \label{eq:subopt-pi-2-ineq}
    &\geq
    \E_{\pi_1\pi_2}[r_{t+1} + \sgn_1(s_t) \sgn_1(s_{t+1}) V(s_{t+1})
    \notag\\&\qquad
    \mid s_t = s], \quad \forall s \in \scr{S}_2, \forall \pi_1, \pi_2.
  \end{align}
  From~\eqref{eq:subopt-pi} with $\pi_2 = \pi_2^*$
  and~\eqref{eq:subopt-pi-2} with $\pi_1 = \pi_1^*$, we obtain
  \begin{align*}
    V(s) &= \E_{\pi_1^*\pi_2^*}[r_{t+1} + \sgn_1(s_t) \sgn_1(s_{t+1})V(s_{t+1})
    \notag\\&\qquad
    \mid s_t = s], \quad \forall s \in \scr{S}
  \end{align*}
  and we conclude from Proposition~\ref{prop:dp-ineq} that
  \begin{align}\label{eq:value-star}
    V(s) = V_{\pi_1^*, \pi_2^*}(s), \quad \forall s \in \scr{S}.
  \end{align}
  Multiplying~\eqref{eq:subopt-pi} by $\sgn_1(s) = 1$,
  $\forall s \in \scr{S}_1$ and combining this equality
  with~\eqref{eq:subopt-pi-2-ineq} multiplied by $\sgn_1(s) = -1$,
  $\forall s \in \scr{S}_2$ and with $\pi_1 = \pi_1^*$, we obtain
  \begin{multline*}
    \sgn_1(s) V(s) \leq \E_{\pi_1^*\pi_2}[r_{t+1} \sgn_1(s_t) \\
    + \sgn_1(s_{t+1}) V(s_{t+1}) \mid s_t = s], \quad \forall s \in \scr{S}.
  \end{multline*}
  In this case, we conclude from~\eqref{eq:value-star} and
  Proposition~\ref{prop:dp-ineq} that
  \begin{multline*}
    \sgn_1(s) V(s) = \sgn_1(s) V_{\pi_1^*, \pi_2^*}(s) \\
    \leq \sgn_1(s) V_{\pi_1^*, \pi_2}(s), \quad \forall s \in \scr{S}.
  \end{multline*}
  and since $\sgn_1(s) = -\sgn_2(s)$, $\forall s \in \scr{S}$, we also
  conclude that
  \begin{multline} \label{eq:value-first}
    \sgn_2(s) V(s)
    = \sgn_2(s) V_{\pi_1^*, \pi_2^*}(s)\\
    \geq \sgn_2(s) V_{\pi_1^*, \pi_2}(s), \quad \forall s \in \scr{S}.
  \end{multline}
  Finally, multiplying~\eqref{eq:subopt-pi-2} by $\sgn_1(s) = -1$,
  $\forall s \in \scr{S}_2$ and combining this equality
  with~\eqref{eq:subopt-pi-ineq} multiplied by $\sgn_1(s) = 1$,
  $\forall s \in \scr{S}_1$ and with $\pi_2 = \pi_2^*$, we obtain
  \begin{multline*}
    \sgn_1(s) V(s) \geq \E_{\pi_1\pi_2^*}[r_{t+1} \sgn_1(s_t) + \\
    \sgn_1(s_{t+1}) V(s_{t+1}) \mid s_t = s], \quad \forall s \in \scr{S}.
  \end{multline*}
  In this case, we conclude from~\eqref{eq:value-star} and
  Proposition~\ref{prop:dp-ineq} that
  \begin{multline} \label{eq:value-second}
    \sgn_1(s) V(s)
    = \sgn_1(s) V_{\pi_1^*, \pi_2^*}(s)\\
    \geq \sgn_1(s) V_{\pi_1, \pi_2^*}(s), \quad \forall s \in \scr{S}.
  \end{multline}
  The saddle-point inequalities~\eqref{eq:saddle-point} follow
  from~\eqref{eq:value-first}, \eqref{eq:value-second}.\frQED
\end{proof}

The following result is just a reformulation of
Lemma~\ref{lemma:dp-saddle} in terms of a so called ``$Q$-function,''
which will provide an explicit formula~\eqref{eq:pi-star-Q} for the
saddle-point policies that can be evaluated without knowing the
transition/reward probability function $p(\cdot|\cdot)$. It will be used
shortly to construct the $Q$-learning algorithm.
\begin{corollary} [Saddle-point sufficient condition using a
  $Q$-function]\label{corr:dp-value} Suppose there
  exists a function $Q : \scr{S} \times \scr{A} \to \mathbb{R}$ for which
  \begin{align}\label{eq:V-from-Q}
    V(s) \eqdef \max_{a \in \scr{A}} Q(s, a), \quad \forall s \in \scr{S},
  \end{align}
  is absolutely summable and
  \begin{multline}\label{eq:maxmin-qfunc}
    Q(s, a) = \E \big[ r_{t+1}
    + \sgn_1(s_t) \sgn_1(s_{t+1}) \max_{a' \in \scr{A}} Q(s_{t+1}, a') \\
    \,\big|\, s_t = s, a_t = a \big], \quad \forall s \in \scr{S},\; a \in \scr{A},
  \end{multline}
  then any pair of policies $(\pi_1^*, \pi_2^*)$ for which
  \begin{align}\label{eq:pi-star-Q}
    \pi_i^*(s) \in \arg\max_{a \in \scr{A}} Q(s, a), \quad \forall s \in \scr{S}_i, i \in \{1,2\}
  \end{align}
  form a (pure) saddle-point and are also security policies.
\end{corollary}
\begin{proof} This result follows directly from the observation that
  the function defined by~\eqref{eq:V-from-Q} satisfies the
  assumptions of Lemma~\ref{lemma:dp-saddle} and leads to the same
  pair of saddle-point policies.\frQED
\end{proof-lemma}

\section{$Q$-learning}\label{sec:q-learning}

$Q$-learning can be used to iteratively construct a function
$Q: \scr{S} \times \scr{A} \rightarrow \mathbb{R}$ that satisfies the fixed-point equation
in~\eqref{eq:maxmin-qfunc}. In the context of zero-sum turn games,
\emph{$Q$-learning} starts from some initial estimate
$\scr{Q}^0:\scr{S} \times \scr{A} \rightarrow \mathbb{R}$ and iteratively draws samples
$(s_t,a_t,s_{t+1},r_{t+1})$ from the transition/reward probability
function $p(s_{t+1},r_{t+1}\mid s_t, a_t)$, each leading to an update of
the form
\begin{align}\label{eq:Q-learning-update-games}
  Q^{k+1}(s_t,a_t)=(1-\alpha_k) Q^k(s_t,a_t)+\alpha_k Q^{k+1}_\mrm{target},
\end{align}
for some sequence $\alpha_k\in(0,1]$ and
\begin{align*}
  Q^{k+1}_\mrm{target}\eqdef r_{t+1}
  +\sgn_1(s_t)\sgn_1(s_{t+1})\max_{a'\in\scr{A}}Q^k(s_{t+1},a').
\end{align*}
The sequence of samples $\{(s_t,a_t,s_{t+1},r_{t+1})\}$ are typically
extracted from traces of multiple episodes of the game. Convergence of
$Q^k$ to a fixed-point of~\eqref{eq:maxmin-qfunc} typically requires
the following \emph{Exploration Assumption} for this sequence:
\begin{assumption}[Exploration]\label{asm:exploration}
  Every pair $(s_t,a_t)\in\scr{S}\times\scr{A}$ appears infinitely often in
  the sequence of samples $\{(s_t,a_t,s_{t+1},r_{t+1})\}$ used in the
  $Q$-learning iteration.\frqed
\end{assumption}
To guarantee convergence, one often also needs the operator from the
space of function $\scr{S}\times\scr{A}\to\R$ to itself, defined by
\begin{multline*}
  Q(s,a) \mapsto \E\Big[r_{t+1}+\\
  \sgn_1(s_t)\sgn_1(s_{t+1})
  \max_{a\in\scr{A}} Q(s_{t+1},a)
  | s_t=s,a_t=a\Big],
\end{multline*}
to be a
contraction~\cite{littman_generalized_1996}. 
This typically holds for discounted costs, but it is somewhat harder
to establish for undiscounted costs like~\eqref{eq:reward-i}
\cite{PatekBertsekas99}. However, for deterministic games that
terminate in finite time, like the ones we discuss in
Section~\ref{sec:numerical-results}, convergence of
\eqref{eq:Q-learning-update-games} can be established without further
assumptions.
\begin{lemma}[Deterministic finite games]
  \label{le:finite-games}
  Assume that the Exploration Assumption~\ref{asm:exploration} holds
  and that the game is deterministic and terminates in finite time.
  Setting $\alpha_k=1$, $\forall k\ge 0$ and initializing
  \begin{align*}
    Q^0(s,a)=0, \quad \forall s\in\scr{S},\;a\in\scr{A},
  \end{align*}
  the sequence of functions $Q^k$ converges in a finite number of
  iterations to a function $Q^*$ for which the sufficient condition
  for optimality~\eqref{eq:maxmin-qfunc} holds.  \frqed
\end{lemma}


The following argument used to prove Lemma~\ref{le:finite-games} bears
similarity to that in~\cite{watkins_technical_1992} for single-player
problems, but the argument we present below is self-contained and
specialized to the setting we consider here.
\begin{proof-lemma}{\ref{le:finite-games}}
  Let $\scr{S}_\mrm{recurrent}\subset\scr{S}$ denote the set of
  \emph{recurrent states}, i.e., the set of states $s\in\scr{S}$ for
  which there is a finite sequence of actions that takes the system
  from the state $s$ to some state in $\scr{S}_\mrm{recurrent}$ with
  probability one. The set
  $\scr{S}_\mrm{transient}\eqdef \scr{S}\setminus\scr{S}_\mrm{recurrent}$ of
  states that are not recurrent are called \emph{transient}. We
  define
  \begin{align*}
    Q^*(s,a)=0, \quad \forall s\in\scr{S}_\mrm{recurrent}, \; a\in\scr{A}.
  \end{align*}
  To complete the construction of the function of $Q^*(\cdot)$ to which
  $Q^k(\cdot)$ converges, consider a directed graph $\scr{G}$ whose nodes
  are the transient states in $\scr{S}_\mrm{transient}$, with an edge
  from $s\in\scr{S}_\mrm{transient}$ to $s'\in\scr{S}_\mrm{transient}$ if
  there is an action $a\in\scr{A}$ for which a transition from $s$ to
  $s'$ is possible, i.e.,
  \begin{align*}
    \exists a\in\scr{A}, r\in\scr{R}: p(s',r|s,a)=1.
  \end{align*}
  This graph cannot have cycles because any node in a cycle would be
  recurrent and thus not in $\scr{S}_\mrm{transient}$. The absence of
  cycles guarantees that $\scr{G}$ has at least one topological
  ordering $\prec$, i.e., we can order the transient states
  $\scr{S}_\mrm{transient}$ so that $s\prec s'$ if and only if the state
  $s'$ can be reached from $s$~\cite{cormen_introduction_2009}.

  \medskip

  First note that every recurrent state $s\in\scr{S}_\mrm{recurrent}$
  must have zero reward, i.e.,
  \begin{align*}
    &p(s',r|s,a)=0, \quad \forall s\in\scr{S}_\mrm{recurrent}, a\in\scr{A},r\neq 0.
  \end{align*}
  This is because all these states can be visited infinitely many
  times, which would contradict the fact that the reward must be zero
  after time $T$. This means that the update
  rule~\eqref{eq:Q-learning-update-games} will keep
  \begin{align}\label{eq:recurrent-zero-games}
    Q^k(s,a)=0, \quad \forall s\in\scr{S}_\mrm{recurrent},\;a\in\scr{A}, \;k\ge 0.
  \end{align}
  Let $a$ be an arbitrary action in $\scr{A}$ and
  $s_{\max}\in\scr{S}_\mrm{transient}$ the ``largest'' state in
  $\scr{S}_\mrm{transient}$ with respect to the order $\prec$, i.e.,
  $s_{\max}\in\scr{S}_\mrm{transient}$ the ``largest'' state in
  $s_{\max}\succ s$, $\forall s\in\scr{S}_\mrm{transient}\setminus \{s\}$.

  \medskip

  In view of the Exploration Assumption~\ref{asm:exploration}, the
  value of $Q(s_{\max},a)$ will eventually be updated at some
  iteration $k$ using~\eqref{eq:Q-learning-update-games}. Moreover,
  since $s_{\max}$ is the ``largest'' state, it cannot transition to
  any transient state so it must necessarily transition to a recurrent
  state. In view of~\eqref{eq:recurrent-zero-games}, the update
  in~\eqref{eq:Q-learning-update-games} with $\alpha=1$ will necessarily be
  of the form
  \begin{align}\label{eq:Q-largest-games}
    Q^{k+1}(s_{\max},a)=r_{t+1}\defeq Q^*(s_{\max},a),
  \end{align}
  where $r_{t+1}$ is the (deterministic) reward arising from state
  $s_t=s_{\max}$ and action $a_t=a$. This means that
  $Q^k(s_{\max},a)$, $\forall a\in\scr{A}$ will converge to the values
  $Q^*(s_{\max},a)$ defined in~\eqref{eq:Q-largest-games}, right after
  their first update.

  \medskip

  We now use an induction argument to show that the same will happen
  for every other $s\in\scr{S}_\mrm{transient}$ and every
  $a\in\scr{S}$. To this effect, consider some
  $s\in\scr{S}_\mrm{transient}$ and assume that at some iteration $k$ at
  which $Q^k(s,a)$ will be updated, $Q^k(s',a)$ has already converged
  to $Q^*(s',a)$, for all $s'\succ s$, $a\in\scr{A}$. At this (and at any
  subsequent update for $s,a$), the update
  in~\eqref{eq:Q-learning-update-games} with $\alpha=1$ will necessarily be
  of the form
  \begin{align*}
    Q^{k+1}(s,a)
    =r_{t+1}+\sgn_1(s)\sgn_1(s_{t+1})\max_{a'\in\scr{A}}Q^k(s_{t+1},a'),
  \end{align*}
  where $r_{t+1}$ and $s_{t+1}\succ s$ are the (deterministic) reward and
  next state, respectively, arising from state $s_t=s$ and action
  $a_t=a$, which by the induction hypothesis means that
  \begin{align}
    \begin{split}
      &Q^{k+1}(s,a)=r_{t+1}+ \\
      &\sgn_1(s)\sgn_1(s_{t+1})\max_{a'\in\scr{A}}Q^*(s_{t+1},a')
      \defeq Q^*(s,a).
    \end{split}
  \end{align}
  This confirms that $Q^k(s,a)$ will converge right after this
  update; completing the induction argument and, recursively, defining
  $Q^*(\cdot)$ for every transient state. \frQED
\end{proof-lemma}

\begin{remark}[Initialization]
  The initialization of $Q^0(s,a)=0$ is actually not needed for
  \emph{every} state $s\in\scr{S}$. An inspection of the proof reveals
  that it is only needed for ``recurrent states,'' i.e., states to
  which the system can return to with probability one by some
  selection of actions.\frqed
\end{remark}

%

\section{Impact of limited computation} \label{sec:fragility}

In games with very large state spaces $\scr{S}$, it is generally not
possible for the Exploration Assumption~\ref{asm:exploration} to hold
for every pair $(s_t,a_t)$ in $\scr{S}\times\scr{A}$. To try to overcome
this difficulty, it is common to progressively decrease exploration
and focus on exploitation. In the context of $Q$-learning,
\emph{exploitation} corresponds to selecting samples
$(s_t,a_t,s_{t+1},r_{t+1})$ with $a_t$ chosen based on the current
estimate $Q^k(\cdot)$, i.e.,
\begin{align} \label{eq:best-resp}
  a_t \in \argmax_{a\in\scr{A}} Q^k(s_t,a).
\end{align}
\emph{Exploration} would instead select the $a_t$ randomly among the
whole set $\scr{A}$ to make sure that all actions and reachable states
are visited, including those that look ``bad'' in view of the current
value of $Q^k(\cdot)$.
By progressively focusing the sampling on exploitation, the $Q$-learning
algorithm can terminate at some finite iteration $K$, with the
Exploration Assumption~\ref{asm:exploration} holding on some subset
$\scr{S}^K\subset\scr{S}$ that is ``invariant'' for the final policies
\begin{align}\label{eq:pi-K}
  \pi_i^K(s) &\in \arg\max_{a \in \scr{A}} Q^K(s, a), \quad \forall s \in \scr{S}_i, i \in \{1,2\},
\end{align}
where ``invariant'' means that the state will remain in $\scr{S}^K$ with
probability 1 under these policies. In practice, this means that the
fixed-point equation~\eqref{eq:maxmin-qfunc} can only hold for states
in $\scr{S}^K$.

In the case of a single player, say player $\P1$ (corresponding to
$\scr{S}_2=\emptyset$), we could adapt the proof of
Proposition~\ref{prop:dp-ineq} to conclude that, when the fixed-point
equation~\eqref{eq:maxmin-qfunc} holds over $\scr{S}^K$, the policy
$\pi_1^K$ in~\eqref{eq:pi-K} would lead to the cost
\begin{align*}
  J_1^K(s_1)=\sum_{t=1}^\infty\E_{\pi_1^K}[r_{t+1}]= \max_{a\in\scr{A}} Q^K(s_1,a).
\end{align*}
This shows that, while $\pi_1^K$ may not be optimal, for a single player
$\P1$, the final $Q^K(\cdot)$ still provides an accurate estimate for the
cost associated with the corresponding policy $\pi_1^K$. However, this
reasoning does not extend to two players, in which case $Q^K(\cdot)$ does
not provide a bound on the security levels~\eqref{eq:security-policy} and
we may have
\begin{align*}
  &
  \sgn_1(s_1)\max_{a\in\scr{A}} Q^K(s_1,a)
  \lesseqgtr\min_{\pi_2} \sgn_1(s_1)V_{\pi_1^K,\pi_2}(s_1)\\
  &
  \sgn_2(s_1)\max_{a\in\scr{A}} Q^K(s_1,a)
  \lesseqgtr\min_{\pi_1} \sgn_2(s_1)V_{\pi_1,\pi_2^K}(s_1).
\end{align*}
The key difficulty in 2-player games is that, even if we get the
fixed-point equation~\eqref{eq:maxmin-qfunc} to hold over a subset
$\scr{S}^K$ of $\scr{S}$ that is invariant for the final policies in
\eqref{eq:pi-K}, the opponent may discover policies that take the
state outside $\scr{S}^K$. Opponents can then exploit this to achieve
rewards not ``accounted for'' in $Q^K(\cdot)$. In practice, this means
that the policies~\eqref{eq:pi-K} may have weaknesses that cannot be
inferred from $Q^K(\cdot)$.

\medskip

The following result shows how one can ``protect'' the policy
\eqref{eq:pi-K} against a fixed set of policies by guaranteeing that
$Q$-learning explores a sufficiently large subset of the state-space
$\scr{S}$. To formalize this result we need the following definition:
given sets of policies $\Pi_1=\{\pi_1:\scr{S}_1\to\scr{A}\}$ and
$\Pi_2=\{\pi_2:\scr{S}_2\to\scr{A}\}$ for players $\P1$ and $\P2$,
respectively, we say that a subset $\scr{S}^\perp$ of $\scr{S}$ is
\emph{invariant for $\Pi_1$, $\Pi_2$} if for every initial state $s_1$ in
$\scr{S}^\perp$, all subsequent states remain in $\scr{S}^\perp$ with
probability one for all the policies in $\Pi_1$, $\Pi_2$.
\begin{corollary} \label{corr:restricted-security} Suppose there
  exists a function $Q : \scr{S} \times \scr{A} \to \mathbb{R}$ and a set
  $\scr{S}^\perp\subset\scr{S}$ for which~\eqref{eq:V-from-Q} is absolutely
  summable and~\eqref{eq:maxmin-qfunc} holds
  $\forall s\in \scr{S}^\perp$, $a \in \scr{A}$. If $\scr{S}^\perp$ is invariant for
  sets of policies $\Pi_1$, $\Pi_2$ that include the \emph{exploitation}
  policies
  \begin{align}\label{eq:pi-perp-Q}
    \pi_i^\perp(s) \in \arg\max_{a \in \scr{A}} Q(s, a), \quad \forall s \in \scr{S}_i, i \in \{1,2\}
  \end{align}
  then $(\pi_1^\perp, \pi_2^\perp)$ are security policies restricted to the sets
  of policies $\Pi_1$, $\Pi_2$, with values $\sgn_1(s) V(s)$ and
  $\sgn_2(s) V(s)$ for the players $\P1$ and $\P2$, respectively;
  specifically
  \begin{align} \label{eq:rest-values}
    \sgn_1(s) V(s)
    = \min_{\pi_2\in\Pi_2} \sgn_1(s_1) V_{\pi_1^\perp, \pi_2}(s_1)\\
    \sgn_2(s) V(s)
    = \min_{\pi_1\in\Pi_1} \sgn_2(s_1) V_{\pi_1^\perp, \pi_2}(s_1).
    \tag*\frqed
  \end{align}
\end{corollary}

\begin{proof}
  To prove this result, we first observe that
  Proposition~\ref{prop:dp-ineq} still holds
  if~\eqref{eq:DP-policy-sufficiency} only holds over
  $\forall s\in\scr{S}^\perp$, as long as $\scr{S}$ is invariant for the policies
  $\pi_1$, $\pi_2$. This enable us to ``reuse'' the steps of proof of
  Lemma~\ref{lemma:dp-saddle} that were used to
  derive~\eqref{eq:value-first} and~\eqref{eq:value-second}, as long
  as we restrict our attention to policies in $\Pi_1$ and
  $\Pi_2$. Specifically, we conclude that
   \begin{align*}
    \sgn_1(s) V(s)
    = \sgn_1(s) V_{\pi_1^\perp, \pi_2^\perp}(s)
    \geq \sgn_1(s) V_{\pi_1, \pi_2^\perp}(s),\\
    \sgn_2(s) V(s)
    = \sgn_2(s) V_{\pi_1^\perp, \pi_2^\perp}(s)
    \geq \sgn_2(s) V_{\pi_1^\perp, \pi_2}(s),
  \end{align*}
  $\forall s\in\scr{S}^\perp$, $\forall \pi_1\in\Pi_1, \pi_2\in\Pi_2$. This shows that the pair
  $(\pi_1^\perp,\pi_2^\perp)$ is a saddle-point \emph{restricted to the sets of
    policies $\Pi_1$, $\Pi_2$} and therefore also security policies also
  restricted to the same set, which means that
  \begin{align*}
    \sgn_1(s_1) V(s)
    = \min_{\pi_2\in\Pi_2} \sgn_1(s_1) V_{\pi_1^\perp, \pi_2}(s_1)\\
    \sgn_2(s_2) V(s)
    = \min_{\pi_1\in\Pi_1} \sgn_2(s_1) V_{\pi_1^\perp, \pi_2}(s_1),
  \end{align*}
  from which the result follows.\frQED
\end{proof}

\section{Exploration policies}\label{sec:exploration}

$Q$-learning is an ``off-policy'' learning algorithm, meaning that
convergence to the fixed point does not rely on the specific algorithm
used to select the actions that appear in the sequence of samples
$\{(s_t,a_t,s_{t+1},r_{t+1})\}$ used for the iteration, as long as the
Exploration Assumption~\ref{asm:exploration} holds. Motivated by
Corollary~\ref{corr:restricted-security}, we now propose an algorithm
that guarantees protection against a given set of policies, when full
exploration is not possible.

\subsection{Opponent-informed exploration}\label{sec:opp-inf}

The proposed \emph{opponent-informed exploration} algorithm starts
with sets of policies $\Pi_1^\mrm{protect}$ and $\Pi_2^\mrm{protect}$, for
players $\P1$ and $\P2$, respectively, against which we would like to
``protect.'' Such policies typically arise from heuristics developed
by expert players.

The $Q$-learning samples $\{(s_t,a_t,s_{t+1},r_{t+1})\}$ are then
collected by simulating many sequences of games, each called an
\emph{episode}. To construct the simulation of the $\ell$th episode, we
pick a pair of policies $(\pi_1^\ell,\pi_2^\ell)$ from the set
\begin{multline}\label{eq:policy-sets}
  \Big\{\pi_1 \in \Pi_1^\mrm{protect}, \pi_2=\pi_2^k\Big\}
  \cup \Big\{\pi_1=\pi_1^k, \pi_2 \in \Pi_2^\mrm{protect} \Big\} \\
  \cup \Big\{ \pi_1 = \pi_1^k, \pi_2 = \pi_2^k \Big\},
\end{multline}
where $\pi_1^k,\pi_2^k$ are the exploitation policies~\eqref{eq:best-resp}
for the current estimate $Q^k(\cdot)$ of the $Q$-function at the start of
the episode $\ell$. Upon simulation of the full episode, the
$Q$-function is updated using~\eqref{eq:Q-learning-update-games} for
all the samples collected during the episode, eventually leading to a
new estimate $Q^{k+L_\ell}(\cdot)$, where $L_\ell$ denotes the number of samples
collected during the $\ell$th episode. The next episode should pick a
different pair of policies in~\eqref{eq:policy-sets}, typically in a
round-robin fashion.

\medskip

To apply Corollary~\ref{corr:restricted-security}, the exit condition
must guarantee that we can only terminate at an iteration $K$ if the
set of states $\scr{S}^K$ visited by all the $Q$-learning samples
$\{(s_t,a_t,s_{t+1},r_{t+1})\}$ up to iteration $K$ is invariant for
the policies in~\eqref{eq:policy-sets}, which include
$\Pi_1^\mrm{protect}$, $\Pi_2^\mrm{protect}$ and the corresponding
exploitation policies~\eqref{eq:best-resp}. In practice, this means
that we can only exit when (i) the set $\scr{S}^K$ stopped growing
(and is thus invariant) and (ii) the function $Q^K(\cdot)$ stopped
changing (so we reached a fixed point within $\scr{S}^K$). For
deterministic games with $\alpha_k=1$ this is straightforward to check,
whereas for stochastic games we will have to ``settle'' for
$Q^K(\cdot)$ changing by less than a smaller tolerance $\epsilon$. We will see in
Section~\ref{sec:numerical-results} that, even for games with very
large state-spaces, this apparently very strict stopping condition can
hold after a relatively small number of episodes, which \emph{explored
  only a very small fraction of the state-space.} The price to pay is
that, while Corollary~\ref{corr:restricted-security} guarantees
protection against $\Pi_1^\mrm{protect}$ and $\Pi_2^\mrm{protect}$, it
neither guarantees protection against all policies, nor does it
guarantee that we cannot find policies that provide ``better''
protection.

\subsection{Tempered exploitation} \label{sec:temp-exp}

Improving the ``quality'' of the policies obtained upon termination of
the $Q$-learning iteration requires enlarging the set of terminal
states $\scr{S}^K$. To accomplish this, we add to the set of policies
in~\eqref{eq:policy-sets}, \emph{tempered exploration Boltzmann
  policies}, which are stochastic policies that pick an action
$a\in\scr{A}$ using the following Boltzmann distribution
\begin{align*}
  \Prob(a_t=a\;|\;s_t=s)=\frac{e^{\beta Q^k(s,a)}}{Q(\beta)}, \quad\forall a\in\scr{A}, s\in\scr{S},
\end{align*}
with $Q(\beta)$ selected so that the probabilities add to one, and the
constant $\beta\ge0$ is called the \emph{temperature parameter.} For
$\beta=0$, all actions are picked with equal probability, but for
$\beta>0$ actions with larger value of $Q^k(s,a)$ are picked with higher
probability. In fact, as $\beta\to\infty$ the Boltzmann policy converges to the
exploitation policy~\eqref{eq:best-resp}. The terminology ``tempered''
refers to the fact that we include in~\eqref{eq:policy-sets} Boltzmann
policies with multiple temperatures from very low values of $\beta$ (pure
exploration) to very high values of $\beta$ (essentially pure
exploitation). This is inspired by the use of multiple temperatures in
the simulation of spin-glass model~\cite{Swendsen1986replica} and
permits a convenient trade-off between exploration and
exploitation. In practice, tempering enlarges the set $\scr{S}^K$ of
visited states, but tempered exploration must eventually be disabled
so that $\scr{S}^K$ stops growing; otherwise the exit condition for
opponent-informed exploration will only hold when the full state-space
has been explored.



\section{Numerical results} \label{sec:numerical-results}

\begin{figure}[t]
  \begin{center}
    \includegraphics[width=6.25cm]{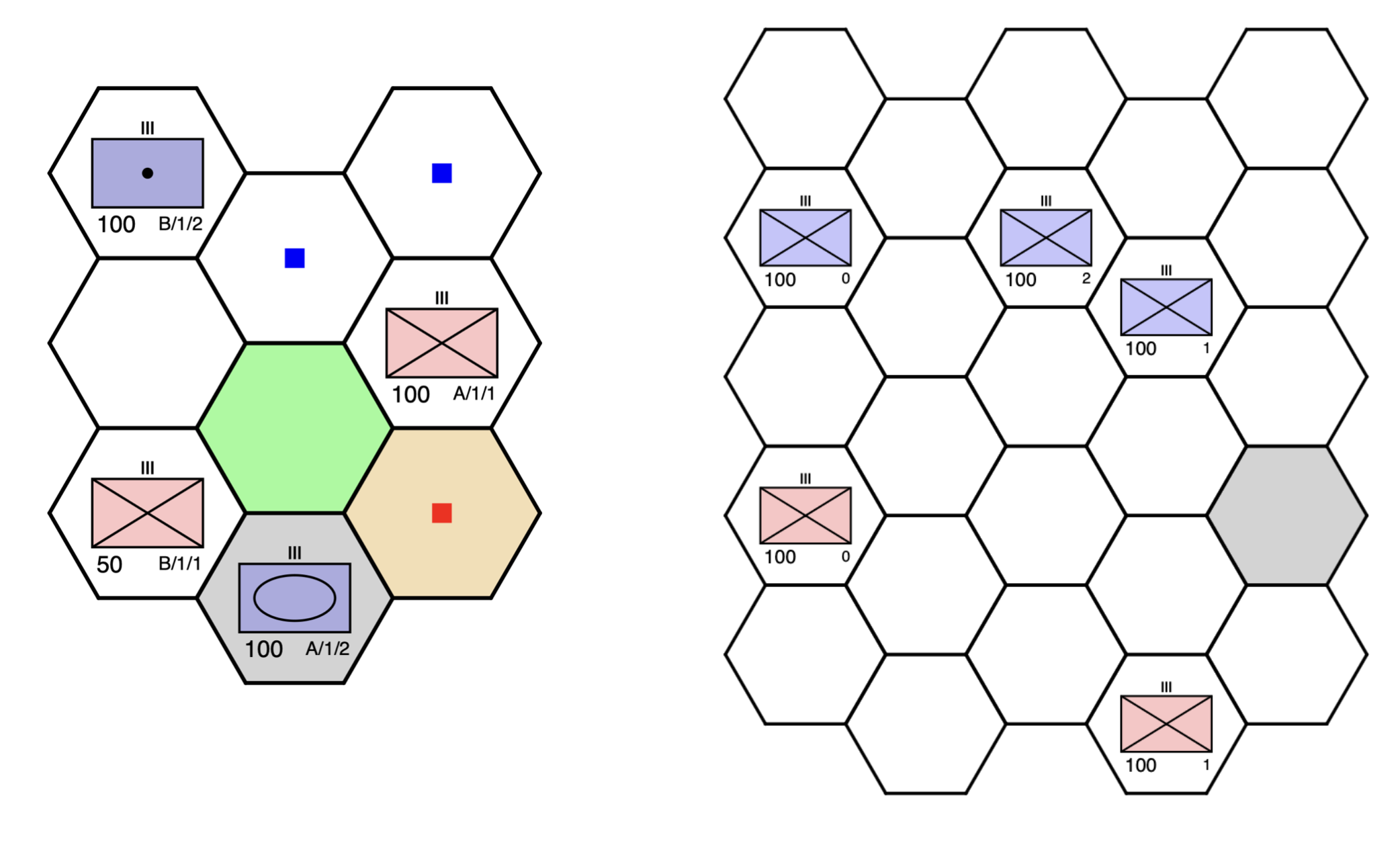} 
    \caption{The left-hand board corresponds to the ``2v2" scenario with one blue artillery (top left), one blue armor (bottom center), and 2 red infantry. The color green indicates a marsh, the beige indicates rough, and grey indicates an urban area. The right-hand side depicts the ``city-inf-5" scenario. While ``2v2" has just over 13 thousand states, ``city-inf-5" has around 1 billion.}
    \label{fig:atlatl}
  \end{center}
\end{figure}

\begin{figure*}[t]
  \centering
  \includegraphics[width=15.8cm]{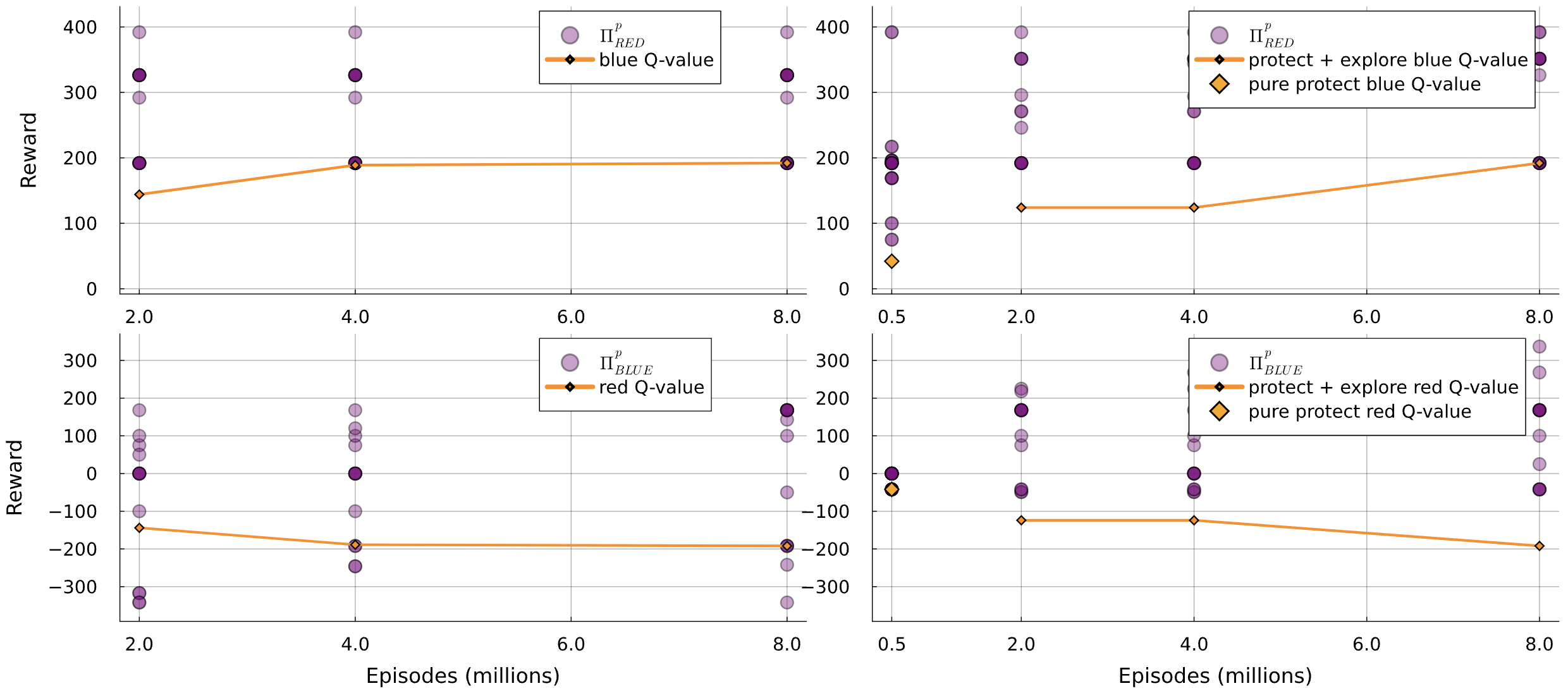} 
  \caption{The left-hand plots show the reward for $Q$-learning as a
    function of training episodes when using a ``unprotected" exploration
    procedure with the top plots showing the blue team's performance
    and the bottom plot showing the red team's. The right-hand plots
    indicate the corresponding performance when using
    opponent-informed exploration to protect against given sets of
    heuristic policies.}
  \label{fig:atlatl-wide}
\end{figure*}

\subsection{Atlatl}
Atlatl is a game implemented in Python by Chris Darken at the Naval Postgraduate School~\cite{cjdarken_public_2025,rood2022scaling}. The
game takes place on a hexagonal grid with two factions, red and blue,
taking turns. A faction consists of multiple units (infantry,
mechanized infantry, armor, artillery), and each unit can move once
per turn, resulting in multiple actions per turn. Distinct units have
different mobility depending on the terrain (clear, rough, marsh,
urban, unused, and water). The color of the hexagon corresponds to the
terrain type. The game state is the full board configuration including
whose turn it currently is. Each unit has a corresponding action space
consisting of either shooting at an opponent or moving to an adjacent
space. Combat is deterministic and rewards are determined according to
a salvo model. Additionally, claiming or holding urban areas
contributes positively to the faction's reward. Each game has a
pre-specified duration. Many heuristic opponent policies are available
in Atlatl, several of which are advanced enough to be hard to beat as
we show later. The game uses OpenAI Gym environment and Stable
Baselines to standardize the interface. A graphic user-interface,
shown in Figure~\ref{fig:atlatl}, also allows for manual game play.

We consider the Atlatl scenario ``city-inf-5'' shown in
Figure~\ref{fig:atlatl}, which contains on the order of one billion
distinct states, making exhaustive exploration computationally
prohibitive. Our simulations ran on the CPU of a 2021 M1 Max chip with
10 cores and 32GB RAM such that each episode and value update
collectively took about 2.5ms. We set
$\Pi_{red}^\mrm{protect}, \Pi_{blue}^\mrm{protect}$ to both consist of 16
known Atlatl heuristic policies that we want to ``protect'' against.

We first consider the performance of $Q$-learning when using tempered
exploration without explicit ``protection''; essentially taking the
sets $\Pi_1^\mrm{protect}$ and $\Pi_2^\mrm{protect}$
in~\eqref{eq:policy-sets} to be
empty. Figure~\ref{fig:atlatl-wide}(left), shows the worst-case empirical reward
as a function of the number of training episodes. We observe that the
$Q$-learning policy for the red team (bottom) is not a security policy
in the sense of~\eqref{eq:security-policy}, because the red player
gets rewards (purple dots) below the $Q$-value (orange) against some
of the 16 Atlatl heuristic policies. For the blue team (top), the
$Q$-learning policies turns out to protect against these specific
heuristics.

We then use opponent-informed exploration from
Section~\ref{sec:opp-inf}, to provide protection against
$\Pi_{red}^\mrm{protect}$ and $\Pi_{blue}^\mrm{protect}$. The algorithm
terminates in just about 500,000 iterations and we depict the rewards
obtained in Figure~\ref{fig:atlatl-wide}(right), with the rewards
appearing at the 500,000 episodes mark. Playing against the policies
in $\Pi_{red}^\mrm{protect}, \Pi_{blue}^\mrm{protect}$ (purple dots)
provides a reward no worse than the $Q$-value (orange line),
confirming the results in Corollary~\ref{corr:restricted-security}:
with the orange line corresponding to the left-hand side
of~\eqref{eq:rest-values} and the purple dots to specific opponent
policies.

Finally, we combine opponent-informed exploration from
Section~\ref{sec:opp-inf} with tempered exploration
from~Section~\ref{sec:temp-exp}. The results are displayed
in~Figure~\ref{fig:atlatl-wide}(right) and show that we can obtain
rewards closer to those corresponding to the ``unprotected''
$Q$-function while still maintaining security with respect to
$\Pi_{red}^\mrm{protect}, \Pi_{blue}^\mrm{protect}$. The price we pay by
``protecting'' is that the saddle-point tends to be more conservative
for the 2 and 4 million episode settings. This is, at least in part,
due to focusing exploration on protecting against the Atlatl
heuristics rather than general exploration.

\section{Conclusion}

This paper presented a computationally efficient $Q$-learning
formulation for zero-sum turn-based Markov games; showed convergence
of it; and proposed an exploration policy that guarantees
``protection'' against a given set of policies for the opponent, when
the $Q$-learning iteration terminates before a full exploration of the
state-space. Our results are independent of how the $Q$-function is
represented but, in practice, the representation used impacts the
value of this function outside the set of states explored by
$Q$-learning. Future work should consider the impact of specific
representation, like the deep neural networks used in deep
$Q$-learning~\cite{mnih_playing_2013}.


%

\bibliographystyle{IEEEtran}
\bibliography{bibliography}

\end{document}